\documentclass[12pt,reqno]{amsart}

\usepackage{vmargin,fancyhdr}   

\usepackage{amsmath}
\usepackage{amsfonts}
\usepackage{hyperref}
\usepackage{amssymb}
\usepackage{color}
\usepackage{graphicx}

\newcommand{\hide}[1]{}

\newcommand{\field}[1]{\mathbb{#1}} 

\newcommand{\beql}[1]{\begin{equation}\label{#1}}
\newcommand{\eeq}{\end{equation}}
\newcommand{\comment}[1]{}

\newcommand{\Mean}[1]{{\mathbb E}\left[{#1}\right]}

\newcommand{\Prob}[1]{{{\bf{Pr}}\left[{#1}\right]}}

\newtheorem{lemma}{Lemma}
\newtheorem{theorem}{Theorem}
\newtheorem{corollary}{Corollary}

\setpapersize{USletter}
\setmarginsrb{.75in}{.5in}        
             {.75in}{.5in}        
             {.25in}{.25in}     
             {.25in}{.5in}      
\begin{document}

\title{The Degree Sequence of Random Apollonian Networks}

\author{Charalampos E. Tsourakakis}
\address{Department of Mathematical Sciences\\
Carnegie Mellon University\\
5000 Forbes Av., 15213\\
Pittsburgh, PA \\
U.S.A} \email{ctsourak@math.cmu.edu}
\thanks{Supported by NSF grant ccf1013110.}

\keywords{Degree Distribution, Random Apollonian Networks}

\maketitle

\begin{abstract}
We analyze the asymptotic behavior of the degree sequence of Random Apollonian Networks \cite{maximal}. 
For previous weaker results see \cite{comment,maximal}.
\end{abstract}

\section{Introduction} 

\subsection{Results}
Random Apollonian Networks (RANs) is a popular model of planar graphs with power law properties \cite{maximal},
see also Section~\ref{sec:model}.
In this note we analyze the degree sequence of RANs. For earlier, weaker results see  \cite{comment,maximal}. 
Our main result is Theorem~\ref{thrm:degrees}. 

\begin{theorem}
\label{thrm:degrees} 

Let $Z_k(t)$ denote the number of vertices of degree $k$ at time $t$, $k \geq 3$.
For any $k \geq 3$ there exists a constant $b_k$ (depending on $k$) such that 
for $t$ sufficiently large

$$ |\Mean{ Z_k(t) } - b_k t| \leq K, \text{~~where~~} K=3.6.$$

\noindent Furthermore let $\lambda > 0$. For any $k \geq 3$

$$\Prob{|Z_k(t) - \Mean{Z_k(t)}| \geq \lambda } \leq 2e^{-\frac{\lambda^2}{72t}}.$$
\end{theorem}

\subsection{Related Work}

Bollob\'{a}s, Riordan,  Spencer and Tusn\'{a}dy \cite{bollobas-degrees} proved rigorously
the power law distribution of the Barab\'{a}si-Albert model \cite{albert}. 
Random Apollonian Networks were introduced in \cite{maximal}. 
Their degree sequence was analyzed inaccurately in \cite{maximal} (see comment in \cite{comment})
and subsequently  in \cite{comment} using physicist's methodology. 
Cooper \& Uehara \cite{cooper} and Gao \cite{gao} analyzed the degree distribution of random $k$ trees, a closely 
related model to RANs. In RANs --in contrast to random $k$ trees-- the random $k$ clique 
chosen at each step has never previously been selected. For example in the two dimensional 
case any chosen triangular face is being subdivided into three new triangular faces by connecting
the incoming vertex to the vertices of the boundary. Darrasse and Soria analyzed the degree
distribution of random Apollonian network structures in \cite{darase}.

\subsection{Model} 
\label{sec:model}

For convenience, we summarize the RAN model here, see also \cite{maximal}.
A RAN is generated by starting with a triangular face and doing the following
until the network reaches the desired size: pick a triangular face uniformly
at random, insert a vertex inside the sampled face and connect it to the 
vertices of the boundary.

\subsection{Prerequisites} 

In Section~\ref{sec:proof} we invoke the following lemma.

\begin{lemma}[Lemma 3.1, \cite{chunglu}]
Suppose that a sequence $\{a_t\}$ satisfies the recurrence 

$$ a_{t+1} = (1 - \frac{b_t}{t+t_1}) a_t + c_t$$

\noindent for $t \geq t_0$.  Furthermore suppose $\lim_{t \rightarrow +\infty} b_t = b >0$
and $\lim_{t \rightarrow +\infty} c_t = c$. Then $\lim_{t \rightarrow +\infty} \frac{a_t}{t}$ exists
and 

$$ \lim_{t \rightarrow +\infty} \frac{a_t}{t} = \frac{c}{1+b}. $$ 

\label{lem:chunglubook} 
\end{lemma}

\noindent We also use in Section~\ref{sec:proof}  the Azuma-Hoeffding inequality~\ref{lem:azuma}, 
see \cite{azuma,hoeffding}.

\begin{lemma}[Azuma-Hoeffding inequality]
Let $(X_t)_{t=0}^n$ be a martingale sequence with $|X_{t+1}-X_t| \leq c$ for 
$t=0,\ldots,n-1$. Also, let $\lambda >0$. Then:
$$ \Prob{ |X_n - X_0| \geq \lambda } \leq 2\exp{ \Big(-\frac{\lambda^2}{2c^2n}\Big) }$$
\label{lem:azuma}
\end{lemma}

\section{Proof of Theorem~\ref{thrm:degrees}}
\label{sec:proof} 

For simplicity, let $N_k(t)=\Mean{Z_k(t)}$, $k \geq 3$.
Also, let $d_v(t)$ denote the degree of vertex $v$ at time $t$
and let $\mathbf{1}(d_v(t)=k)$ be an indicator variable which equals
1 if $d_v(t)=k$, otherwise 0. Then, for any $k\geq 3$ we
can express the expected number $N_k(t)$ of vertices of degree $k$ 
as a sum of expectations of indicator variables:

\begin{equation}
N_k(t) = \sum_{v} \Mean{\mathbf{1}(d_v(t)=k)}.
\label{eq:eqindicators}
\end{equation}
 
\noindent  We distinguish two cases in the following. 
\newline

\noindent \underline{$\bullet$ {\sc Case 1} $k=3$:}\\

\noindent Observe that a vertex of degree 3 is created only by an insertion of a new vertex. 
The expectation $N_3(t)$ satisfies the following recurrence\footnote{The 
three initial vertices participate in one less face than their degree. However,
this leaves the asymptotic analysis unchanged.}

\begin{equation}
N_3(t+1) = N_3(t)+1-\frac{3N_3(t)}{2t+1}.
\label{eq:basis}
\end{equation}

\noindent The basis for Recurrence~\eqref{eq:basis} is $N_3(1)=4$. We prove the following lemma
which shows that $\lim_{t \to +\infty}{\frac{N_3(t)}{t}} = \frac{2}{5}$.

\begin{lemma} 
$N_3(t)$ satisfies the following inequality:
\begin{equation} 
|N_3(t) - \frac{2}{5} t|\leq K, \text{~~where~~} K=3.6
\label{eq:dk3}
\end{equation} 
\label{lem:dk3}
\end{lemma}

\begin{proof} 
We use induction. Assume that $N_3(t)=\frac{2}{5}t+e_3(t)$.
We wish to prove that for all $t$, $|e_3(t)|\leq 3.6$. 
The result trivially holds for $t=1$.
Assume the result holds for some $t$. We shall show it holds for $t+1$.

\begin{align*}
N_3(t+1)  &= N_3(t)+1-\frac{3N_3(t)}{2t+1} \Rightarrow \\
e_3(t+1)    &= e_3(t) + \frac{3}{5} - \frac{6t+15e_3(t)}{10t+5} = e_3(t) + \frac{3}{5(2t+1)} - \frac{3e_3(t)}{2t+1} \Rightarrow \\ 
|e_3(t+1)|  &\leq K(1-\frac{3}{2t+1}) + \frac{3}{5(2t+1)} \leq K \\ 
\end{align*}

Hence by induction, Inequality~\eqref{eq:dk3} holds for all $t \geq 1$. 

\end{proof}

\noindent \underline{$\bullet$ {\sc Case 2} $k \geq 4$:}\\

For $k \geq 4$ the following Equation holds for each indicator variable $\mathbf{1}(d_v(t)=k)$: 

\begin{equation} 
\Mean{\mathbf{1}(d_v(t+1)=k)} = \Mean{ \mathbf{1}(d_v(t)=k)} (1-\frac{k}{2t+1}) + \Mean{\mathbf{1}(d_v(t)=k-1)} \frac{k-1}{2t+1}
\label{eq:degree}
\end{equation}

Therefore,  substituting in Equation~\eqref{eq:eqindicators} the expression
from Equation~\ref{eq:degree} we obtain for $k \geq 4$

\begin{equation}
N_k(t+1) = N_k(t)(1-\frac{k}{2t+1})+ N_{k-1}(t)\frac{k-1}{2t+1}.
\label{eq:eqindicators2}
\end{equation}

Now, we use induction to show that $\lim_{t \to +\infty} \frac{N_k(t)}{t}$ for $k \geq 4$ exists.

\begin{lemma} 
For $k \geq 3$ $\lim_{t \to +\infty} \frac{N_k(t)}{t}$ exists.
Let $b_k=\lim_{t \to +\infty} \frac{N_k(t)}{t}$. Then,
$b_3=\frac{2}{5}, b_4=\frac{1}{5}, b_5=\frac{4}{35}$ and for $k\geq 6$
$b_k = \frac{24}{k(k+1)(k+2)}$.
Furthermore, for all $k \geq 3$ 

\begin{equation} 
|N_k(t) - b_k t| \leq K, \text{~~where~~} K=3.6.
\label{eq:dk4}
\end{equation}
\label{lem:degreelemma}

\end{lemma}

\begin{proof} 
For $k=3$ the result holds by Lemma~\ref{lem:dk3}. We use induction. Rewrite Recurrence~\eqref{eq:eqindicators2} as: 
$ N_k(t+1) = (1- \frac{b_t}{t+t_1})N_k(t) + c_t$ where $b_t=k/2$, $t_1=1/2$,
$c_t = N_{k-1}(t)\frac{k-1}{2t+1}$. Clearly, $\lim_{t \rightarrow +\infty}{b_t}=k/2>0$
and by the inductive hypothesis $\lim_{t \rightarrow +\infty}{c_t}= \lim_{t \rightarrow +\infty}{ b_{k-1}t \frac{k-1}{2t+1}} = b_{k-1}(k-1)/2$. Hence, by invoking Lemma~\ref{lem:chunglubook} we obtain

$$ b_k = \lim_{t \rightarrow +\infty} \frac{N_k(t)}{t} = \frac{ (k-1)b_{k-1}/2 }{1+k/2}= b_{k-1} \frac{k-1}{k+2}. $$ 

Therefore $b_3=\frac{2}{5}$, $b_4=\frac{1}{5}$, $b_5=\frac{4}{35}$
for any $k \geq 6$, $b_k = \frac{24}{k(k+1)(k+2)}$ which shows a power law
degree distribution with exponent 3. 

Finally consider the proof of Inequality~\eqref{eq:dk4}.
The case $k=3$ was proved in the Lemma~\ref{lem:dk3}. Assume the result holds for some $k \geq 3$, i.e., $|e_k(t)| \leq K$ 
where $K=3.6$. We will show it holds for $k+1$ too. Let $e_k(t)=N_k(t)-b_kt$. Substituting in Recurrence~\eqref{eq:eqindicators2}
and using the fact that $b_{k-1}(k-1)=b_k(k+2)$ we obtain the following: 

\begin{align*}
e_k(t+1) &= e_k(t) + \frac{k-1}{2t+1} e_{k-1}(t) - \frac{k}{2t+1} e_k(t) \Rightarrow \\
|e_k(t+1)| &\leq | (1-\frac{k}{2t+1}) e_k(t) | + | \frac{k-1}{2t+1} e_{k-1}(t) | \leq K(1-\frac{1}{2t+1}) \leq K
\end{align*}

Hence by induction, Inequality~\eqref{eq:dk4} holds for all $k \geq 3$.
\end{proof}

It's worth pointing out that Lemma~\ref{lem:degreelemma} agrees 
with \cite{friezetsourakakis} where it was shown that the maximum
degree is $\Theta(\sqrt{t})$. 
Finally, the Lemma~\ref{lem:deg} proves the concentration of $Z_k(t)$ around its expected
value for $k \geq 3$. This lemma applies the 
Azuma-Hoeffding inequality  \ref{lem:azuma} and 
completes the proof of Theorem~\ref{thrm:degrees}.

\begin{lemma} 
Let $\lambda >0$. For any $k\geq 3$

\begin{equation}
\Prob{|Z_k(t) - \Mean{Z_k(t)}| \geq \lambda } \leq 2e^{-\frac{\lambda^2}{72t}}.
\end{equation}

\label{lem:deg} 
\end{lemma}

\begin{proof} 

Let $(\Omega, \mathcal{F}, \field{P})$ be the probability space induced
by the construction of a Random Apollonian Network (see Section~\ref{sec:model}) 
after $t$ insertions. Fix $k$  ($k \geq 3$) and let $(X_i)_{i \in \{0,1,\ldots,t\}}$
be the martingale sequence defined by $X_i = \Mean{Z_k(t)|\mathcal{F}_i}$,
where $\mathcal{F}_0=\{\emptyset,\Omega\}$ and 
$\mathcal{F}_i$ is the $\sigma$-algebra generated by the RAN process 
after $i$ steps. Notice $X_0=\Mean{Z_k(t)|\{\emptyset,\Omega\}}=N_k(t)$, $X_t = Z_k(t)$.
We show that $|X_{i+1}-X_{i}| \leq 6$ for $i=0,..,t-1$. 
Let $P_j=(Y_1,\ldots,Y_{j-1},Y_j)$, $P_j'=(Y_1,\ldots,Y_{j-1},Y_j')$  be two sequences of 
face choices differing only at time $j$. Also, let $\bar{P},\bar{P'}$ continue from $P_j,P_j'$ 
until $t$. We call the faces $Y_j,Y_j'$ special with respect to $\bar{P},\bar{P'}$. 
We define a measure preserving map $\bar{P} \mapsto \bar{P'}$ in the following way:
for every choice of a non-special face in process $\bar{P}$ at time $l$ we make the same face choice in $\bar{P'}$ 
at time $l$. For every choice of a face inside the special face $Y_j$ in process $\bar{P}$ 
we make an isomorphic (w.r.t., e.g., clockwise order and depth) 
choice of a face inside the special face $Y_j'$ in process $\bar{P}'$.
Since the number of vertices of degree $k$ can change by at most 6, i.e., the (at most) 6 vertices involved
in the two faces $Y_j,Y_j'$ the following holds: 

$$ |\Mean{Z_k(t)|P} - \Mean{Z_k(t)|P'}| \leq 6 .$$

Furthermore,  this holds for any $P_j,P_j'$. 
We deduce that $X_{i-1}$ is a weighted mean of values, whose pairwise differences are all at most 6. 
Thus, the distance of the mean $X_{i-1}$ is at most 6 from each of these values.
Hence, for any one step refinement  $|X_{i+1}-X_{i}| \leq 6$ $\forall i\in \{0,\ldots,t-1\}$. 
By applying the  Azuma-Hoeffding inequality as stated in Lemma~\ref{lem:azuma} we obtain

\begin{equation}
\Prob{|Z_k(t) - \Mean{Z_k(t)}| \geq \lambda } \leq 2e^{-\frac{\lambda^2}{72t}}.
\end{equation}

\end{proof}

\noindent A corollary immediately obtained by the previous lemma by setting $\lambda=\sqrt{t\log{t}}$  is the following:

\begin{corollary}
$\Prob{|Z_k(t) -  \Mean{Z_k(t)} | \geq \sqrt{t\log{t}}} = o(1)$.
\end{corollary}

\section*{Acknowledgements} 

\noindent The author would like  to thank Deepak Bal and Alan Frieze for helpful comments on the manuscript.


\begin{thebibliography}{99}

\bibitem{azuma}
Azuma, K.:
{\em Weighted sums of certain dependent variables}
Tohoku Math J 3, pp. 357-367, 1967 



\bibitem{albert}
Barab\'{a}si, A., Albert, R.:
{\em Emergence of Scaling in Random Networks}
Science 286, pp. 509-512, 1999

\bibitem{bollobas-degrees}
Bollob\'{a}s, B., Riordan, O., Spencer, J., Tusn\'{a}dy, G.:
{\em The Degree Sequence of a Scale Free Random Graph Process}
Random Struct. Algorithms, 18(3), pp. 279-290, 2001



\bibitem{cooper}
Cooper, C., Uehara, R.: 
{\em Scale Free Properties of random $k$-trees}
Mathematics in Computer Science, 3(4), pp. 489--496, 2010


\bibitem{chunglu}
Chung Graham, F., Lu, L.: 
{\em Complex Graphs and Networks}
American Mathematical Society, No. 107 (2006)

\bibitem{darase}
Darrasse, A., Soria, M.:
{\em Degree distribution of random Apollonian network structures and Boltzmann sampling}
2007 Conference on Analysis of Algorithms, AofA 07, DMTCS Proceedings.


\bibitem{friezetsourakakis} 
Frieze, A., Tsourakakis, C.E.:
{\em High Degree Vertices, Eigenvalues and Diameter of Random Apollonian Networks}
Available at Arxiv \url{http://arxiv.org/abs/1104.5259}

\bibitem{gao}
Gao, Y.:
{\em The degree distribution of random k-trees}
Theoretical Computer Science, 410(8-10), 2009.


\bibitem{hoeffding}
Hoeffding, W.:
{\em Probability inequalities for sumes of bounded random variables}
J Amer Statist Assoc 58, pp. 13-30, 1963

\bibitem{comment}
Wu, Z.-X., Xu, X.-J., Wang, Y.-H.:
{\em Comment on ``Maximal planar networks with large clustering coefficient and power-law degree distribution''}
Physical Review, E 73, 058101 (2006)

\bibitem{maximal}
Zhou, T., Yan, G., Wang, B.H:
{\em Maximal planar networks with large clustering coefficient and power-law degree distribution}
Physical Review, E 71, 046141 (2005)

\end{thebibliography}
\end{document}